\pdfoutput=1
\documentclass[a4paper,10pt]{article}
\usepackage{amsmath,amssymb}
\usepackage[pdftex]{graphicx}
\usepackage{here}

\usepackage{comment} 

\renewcommand{\title}[1]{\vspace{\fill}
\eject\addtolength{\baselineskip}{4pt}
{\bfseries\LARGE #1}\\[3mm]\addtolength{\baselineskip}{-4pt}}
\renewcommand{\author}[3]{\parbox[t]{75mm}
{\begin{center}{\scshape #1}\\[3mm] #2\\
 {\ttfamily #3} \end{center}}}

\newtheorem{thm}{\bfseries Theorem}[section]
\newtheorem{lem}[thm]{\bfseries Lemma}        
\newtheorem{prop}[thm]{\bfseries Proposition} 
\newtheorem{cor}[thm]{\bfseries Corollary}

\numberwithin{equation}{section}

\newenvironment{proof}{\medskip                    
\noindent{\scshape Proof:}}{\quad $\Box$\medskip}  

\begin{document}

\begin{center}

\title{On the Number of Maximal Cliques in\\
Two-Dimensional Random Geometric Graphs:\\
Euclidean and Hyperbolic} 
\author{
\underline{Hodaka Yamaji}
}{
Graduate School of Information Science and Technology\\
The University of Tokyo\\
Tokyo 113-8656, Japan
}{
    hodakaymj@g.ecc.u-tokyo.ac.jp
}

\end{center}


\begin{quote}
{\bfseries Abstract:}
Maximal clique enumeration appears in various real-world networks, such as social
networks and protein-protein interaction networks for different applications.
For general graph inputs, the number of maximal cliques can be up to $3^{|V|/3}$.
However, many previous works suggest that the number is much smaller than that on real-world networks, and polynomial-delay algorithms enable us to enumerate them in a realistic-time span.
To bridge the gap between the worst case and practice, we consider the number of maximal cliques in two 
popular models
of real-world networks:
Euclidean random geometric graphs and hyperbolic random graphs.
We show that the number of maximal cliques on Euclidean random geometric graphs is lower and upper bounded by
$\exp(\Omega(|V|^{1/3}))$ and $\exp(O(|V|^{1/3+\epsilon}))$ with high probability for any $\epsilon > 0$.
For a hyperbolic random graph, we give the bounds of
$\exp(\Omega(|V|^{(3-\gamma)/6}))$ and $\exp(O(|V|^{(3-\gamma+\epsilon)/6)}))$ where $\gamma$ is 
the power-law degree exponent between 2 and 3.
\end{quote}

\begin{quote}
\end{quote}
\vspace{5mm}



\section{Introduction}
\subsection{Background}
Detecting all maximal cliques in a graph is a crucial analysis tool for real-world networks from
various fields: social networks,
protein-protein interaction networks, and web graphs because cliques correspond to 
meaningful components in the networks \cite{cliquesocial,cliqueprotein,cliqueweb}.
Not only does it have many direct applications, but its algorithms and techniques are used in other 
clique-related methods
such as clique percolation \cite{cliquepercolation} and $k$-clique counting \cite{pivoting}. This is because we can detect all cliques by enumerating only the maximal ones.\par
For general graph inputs, the number of maximal cliques $\mathcal{M}$ can be up to $3^{|V|/3}$ \cite{moonmoser}. Therefore, enumerating all of them is NP-hard. However, many studies report that in real-world networks, $\mathcal{M}$ is much smaller than that. Thus, polynomial-delay algorithms, the running time of which is bounded by $poly(|V|)\cdot \mathcal{M}$, can enumerate all maximal cliques in realistic-time span even for networks with millions of vertices \cite{polynomialdelay}.
Also, classic Bron-Kerbosch algorithm \cite{bron-kerbosch} (plus graph orientation \cite{orientation}) is known to be efficient in many instances \cite{MCE}, although its worst running time is $O^*(3^{|V|/3})$ and not bounded in terms of $\mathcal{M}$.
Here we strike upon the question:
\textbf{why is the number of maximal cliques small on real-world networks?}\par

In the study of real-world networks, networks that appear naturally in various fields are considered. In terms of the graph structure, it seems that networks from different domains are entirely different from each other. However, it is known that they share specific common properties.
For example, they have a power-law degree distribution:
the number of nodes with a vertex degree of $k$ is proportional to $k^{-\gamma}$. 
In many cases, $\gamma$ is between two and three, and these networks are called scale-free. 
Additionally, they have the triadic closure property,  
meaning that if two vertices have common neighbors, they are likely to be connected. 
The property is often described with a measure called the clustering coefficient, 
and real-world graphs often have a high clustering coefficient. Other common properties include tree-like structures, small diameter, and small clique number.

One of the most combinatorially studied models of real-world networks 
is hyperbolic random graphs \cite{hyperbolicgraph}. The graph is generated by independently placing vertices according to 
a particular distribution in a two-dimensional space with negative curvature and connecting two vertices 
within a certain distance. It is thus classified as a random geometric graph. It is known that a hyperbolic plane naturally induces power-law degree distribution \cite{hyperbolicgraph}. Also, a hyperbolic random graph satisfies a high clustering coefficient with high probability \cite{clustering}, which distinguishes itself from other power-law models such as Barab\'{a}si-Albert \cite{barabasialbert} and Chung-Lu random graphs. Parameters studied in this model include:
the number of $k$-cliques, clique number \cite{hyperbolicclique}, treewidth \cite{hyperbolictreewidth}, 
modularity \cite{hyperbolicmodularity}, and diameter \cite{hyperbolicdiameter}. \par
\subsection{Our Main Results and Contributions}
In this paper, we consider the number of maximal cliques in hyperbolic random graphs. We 
also consider the number on two-dimensional Euclidean random geometric graphs, which are the Euclidean counterpart
to hyperbolic random graphs. Euclidean random geometric graphs are also thought to be good representations of
some types of real-world graphs \cite{wireless, randomprotein}, although they do not possess power-law degree distribution.
Our findings are as follows:
\begin{thm} [Main 1]
    Let $r<1$ be a constant.
    Let $\mathcal{M}$ be the number of maximal cliques in a two-dimensional Euclidean random geometric graph whose connection distance is $r$.
    There exists positive constants $C_1$ and $C_2$ such that for all $\epsilon > 0$,
    \begin{align*}
        \Pr[\exp(C_1 |V|^{1/3})\leq \mathcal{M} \leq \exp(C_2 |V|^{1/3+\epsilon})] \rightarrow 1 \\
    \end{align*}
    as $|V| \rightarrow \infty$.
\end{thm}
\begin{thm} [Main 2]
    Let $\gamma \in (2,3)$.
    Let $\mathcal{M}$ be the number of maximal cliques in a hyperbolic random graph 
    whose power-law degree exponent is $\gamma$.
    There exist positive constants $C_1$ and $C_2$ such that for all $\epsilon > 0$,
    \begin{align*}
        \Pr[\exp(C_1 |V|^{(3-\gamma)/6})\leq \mathcal{M} \leq \exp(C_2 |V|^{(3 - \gamma)/6+\epsilon})] \rightarrow 1 \\
    \end{align*}
    as $|V| \rightarrow \infty$.
\end{thm}
For hyperbolic random graphs, we consider the case when $\gamma \in (2,3)$ for hyperbolic random graphs. 
In this case, the graphs are scale-free and have $\exp(\Omega(|V|^{(3-\gamma)/2}))$ cliques with high probability \cite{hyperbolicclique}.
In general graphs, the number of maximal cWe and the bound
we obtained is much smaller than that. \par
To prove the main theorem, we consider what is called an octahedral graph $O_t$. The definition of the graph is the following.
Let $tK_2=(V,E)$ where $V=\{1,2,...,2t\}$ and $E=\{(i,i+t): 1 \leq i \leq t\}$. Therefore, $tK_2$ is a graph
with $t$ pairwise disjoint edges. An octahedral graph $O_t$ is the complement of $tK_2$. We have the following
theorems from the previous study.
\begin{thm} [Forklore]
    \label{maximallower}
    If the graph has $O_t$ as a vertex-induced subgraph, then the number of maximal cliques
    is lower-bounded by $2^t$.
\end{thm}
\begin{thm} [M. Farber, M. Hujter, and Z. Tuza \cite{Ot}]
    \label{maximalupper}
If $|V| \geq 4t$ and there exists 
no $O_{t+1}$ as a vertex-induced subgraph, then the number of maximal cliques is upper-bounded by $(|V|/t)^{2t}$.\par
\end{thm}
Let $\tau(G)$ be the maximum $t$ such that a graph $G$ has $O_t$ as its vertex-induced subgraph.
With these theorems, all that remains is to bound $\tau$.
If $\tau$ is a constant, then the number of maximal cliques is polynomial. Unfortunately, this is not the case for the two random geometric graphs.
However, our bounds on $\tau$ are much smaller than the obvious $O(|V|)$ bound.

Intuitively, $\tau$ is small because
any pair of unconnected vertices have $2t-2$ common neighbors, which is against the triadic closure property, i.e. two vertices with common neighbors are likely to be connected. 
As $t$ gets larger, the more severe the violation becomes.
This explanation can be mathematically justified on Euclidean and hyperbolic random geometric graphs. The 
arguments on the two different random graphs are basically the same even though the definitions of distance are 
different, and the parallel postulate does not hold on a hyperbolic plane. \par
Our contributions are briefly summarized as follows.
Firstly, to the best of our knowledge, this is the first work that assesses the number of maximal cliques
on random geometric graphs. We shed light on the importance of $O_t$ and develop geometric and probabilistic techniques to determine its
size. Those techniques apply to both Euclidean and hyperbolic planes.
Secondly, what we have found is yet another result followed by $c$-closed graphs \cite{closedness}
supporting that the triadic closure property plays an essential role
in maximal cliques. Lastly, we give an upper bound of $\tau$ on real-world graph datasets.
It turns out that $\tau$ is often at most 5-20, even on networks with hundreds of thousands of vertices (See Table 1 at the end of this paper).
The upper bound of $\mathcal{M}$ given by $\tau$ is still far from the actual value.
However, it is still surprising how small $\tau$ is in practice.
\par
\subsection{Proof Sketch}
Here we explain how we bound $\tau$. For the lower bound, we construct regions so that vertices on them would form $O_t$. Here we use a technique called Poissonization introduced in \cite{penrose}. \par
To derive the upper bound, We define a set of segments $\mathcal{S}:=\{\overline{vw}: v,w\in V(G),$ the distance between $v$ and $w$ is greater than the connection distance$\}$. Since $\mathcal{S}$ is somewhat like the complement of the graph, a vertex-induced subgraph that is isomorphic to $O_t$ corresponds to $t$ segments that are pairwise disjoint (or we call ``independent" in this paper). We soon find out that any two independent segments must intersect. Therefore, we can define an ``angle" between them. By the pigeonhole principle, if $t$ is large enough, then there exists 
a set of segments with decent cardinality whose ``angles" are small (Lemma \ref{pigeon} and Lemma \ref{hyperpigeon}).
We also prove that if two intersecting segments have small ``angles", the geometry 
is very restricted. If we fix one segment, then endpoints of the other segments can only 
exist in small regions (Corollary \ref{regionU} and Corollary \ref{hyperregionU}).
By combining the two facts, we derive that $t$ cannot get too large because the number of vertices
in those small regions cannot be too large with high probability.\par
Unlike Euclidean random geometric graphs, vertices in hyperbolic random graphs are placed by non-uniform distribution. To deal with this, we classify the segments by the density around their endpoints. We treat the segments in dense regions slightly differently from those in sparse regions. However, the overall strategies of the proof are the same as Euclidean random geometric graphs.
\begin{figure}[h]
\centering
\begin{center}
    \includegraphics[width = 6cm]{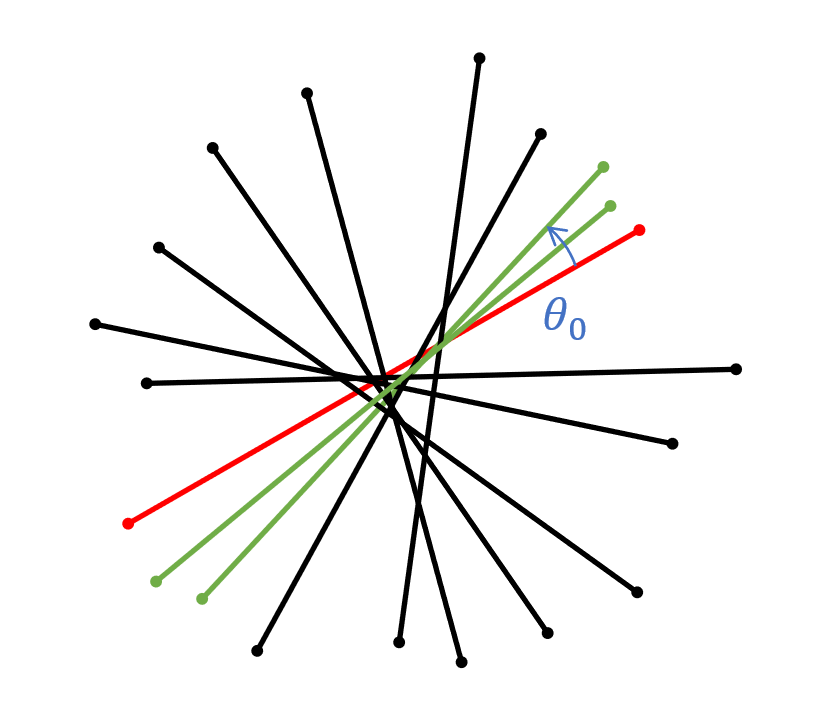}
    \includegraphics[width = 6cm]{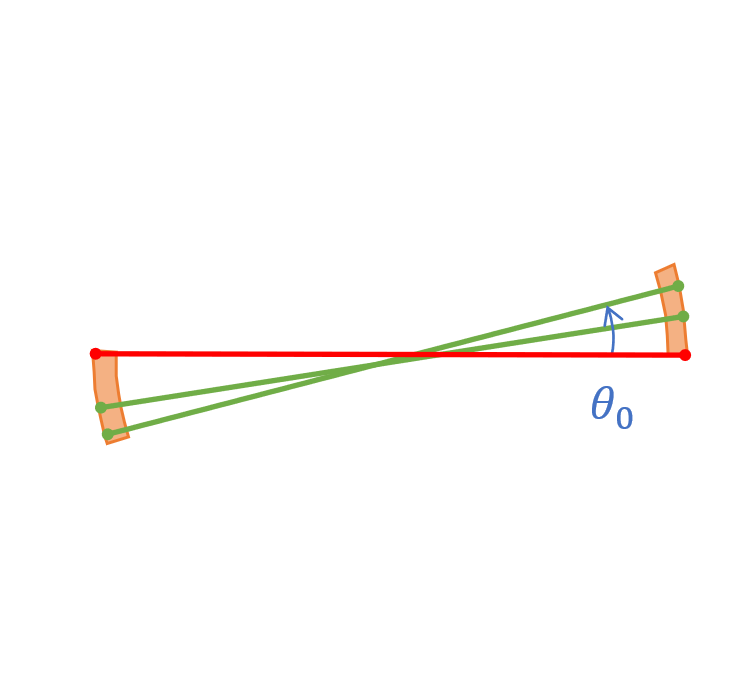}
\end{center}
\caption{Illustrations of the proof sketch. The left one illustrates $t$ independent segments, corresponding to $O_t$ and the claim of Lemma \ref{pigeon} and Lemma \ref{hyperpigeon}. The right one illustrations Corollary \ref{regionU} and Corollary \ref{hyperregionU}.}
\end{figure}

\subsection{Organization of the Paper}
In Section 2, we define Euclidean random geometric graphs and prove the main theorem.
In Section 3.1-3.3, we define hyperbolic random graphs and discuss how the proof of the random geometric graphs on 
a Euclidean plane can be extended to a hyperbolic plane. For the upper bound, we prove the weaker version of the main theorem. In Section 3.4, we discuss how to treat the non-uniformity of hyperbolic random graphs and prove the complete version of the main theorem. Finally, Section 4 is the conclusion.

\section{Euclidean Random Geometric Graphs}
\subsection{Definitions}
Let $n \in \mathbb{N}^+$, and $r \in (0,1)$. A two-dimensional Euclidean random geometric graph $G_{n,r}$ is obtained
as below:
\begin{itemize}
    \item The vertex set is $V =\{1,2,...,n\}$.
    \item The vertices are identically and independently distributed on $[0,1]^2$ according
    to a probability density function $f(x,y)=1$.
    \item The edge set $E$ is given by $\{(u,v):dist(u,v)\leq r\}$
\end{itemize}
Here, $dist(u,v)=\sqrt{(u_x-v_x)^2+(u_y - v_y)^2}$ where $(u_x,u_y)$ and $(v_x,v_y)$ are the $xy$-coordinates of $u$ and $v$ respectively. From here, we often identify a vertex with its position. \par
Given a region $U$ on $[0,1]^2$ (where we can perform integration), define $F(U):=\int_U f(x,y)dx dy$. $F(U)$ is equal to the probability that
a vertex lies on $U$. Note that we can assume no three vertices lie on a single line since the probability is 0. This avoids some corner case arguments.
\subsubsection{Poissonization}
This is introduced in \cite{penrose}.
Let $N_n$ be a Poisson random variable with mean $n$. Consider $G_{N_n,r}$ where the number of vertices is chosen by the Poisson distribution $N_n$. Then, on such a random graph, the number of vertices on a region $U$ follows a Poisson distribution with mean $n F(U)$. In particular, the probability that $U$ has no vertex on itself is $e^{-nF(U)}$. Moreover, for disjoint regions $U_1$ and $U_2$, the number of vertices on each region is an independent Poisson distribution with mean $n F(U_1)$ and $n F(U_2)$, respectively. Therefore, on Poissonized random geometric graphs, arguments of multiple disjoint regions are easier. Let $A$ be some property of a graph. Then
\begin{align*}
    \Pr[G_{N_n,r} \text{ has } A]&\geq \Pr[G_{N_n,r} \text{ has } A | N_{n}=n] \cdot \Pr[N_{n}=n]\\
    &= \Pr[G_{n,r} \text { has } A] \cdot \frac{e^{-n}n^n}{n!}
\end{align*}
Here, $e^{-n}n^n/n! = \Theta(n^{-1/2})$ by the Stirling's approximation. Thus, we get
\begin{align}
    \label{poisson}
    \Pr[G_{n,r} \text { has } A] \leq \Theta(n^{1/2}) \cdot \Pr[G_{N_n,r} \text { has } A]
\end{align}
Although our results are on the normal random geometric graph $G_{n,r}$, we use this technique to derive the lower bound of the number of maximal cliques.

\subsection{Lower Bound of the Number of Maximal Cliques}
Let $k \geq 4$ be an integer. Let $o'=(1/2,1/2)$. 
Consider taking a polar coordinate system whose origin is $o'$.
Let $\theta_0:=\pi/(3k)$.
For $1\leq i \leq 2k$, Define $U_i:=\{(r',\phi): r_1 \leq r' \leq r_2,\ 3(i-1)\theta_0 \leq \phi \leq (3(i-1)+1)\theta_0\}$ where
$r_1:=r/\sqrt{2+2 \cos(\theta_0/2)}$ and $r_2:=r/\sqrt{2+2 \cos (2\theta_0)}$. With some calculations, we can confirm the following.

\begin{figure}[H]
\centering
\begin{center}
    \includegraphics[width = 8cm]{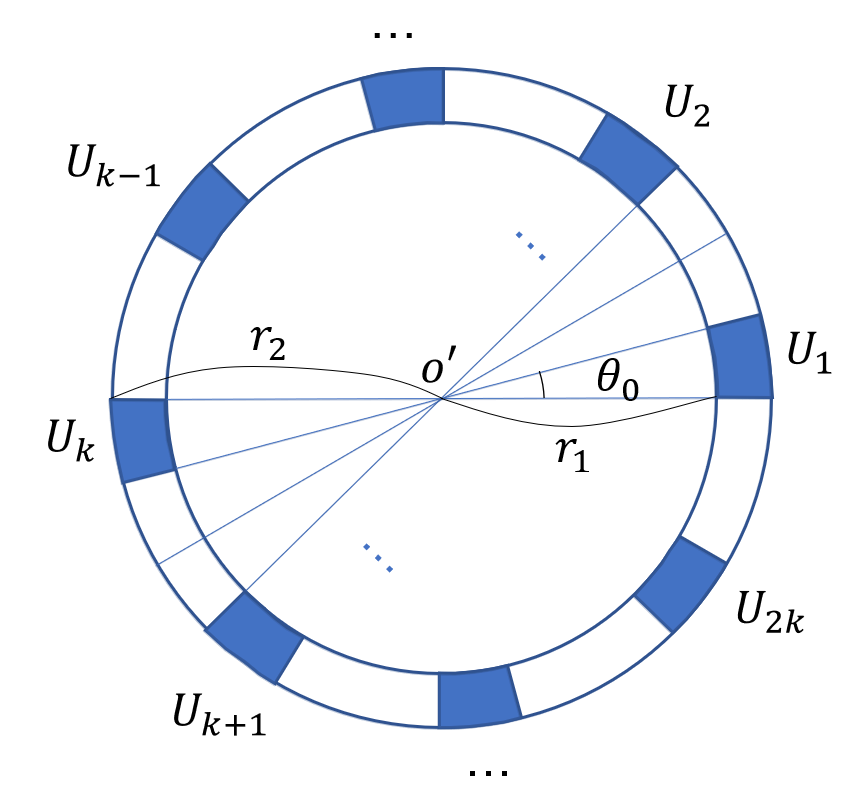}
\end{center}
\caption{Definition of $r_1$, $r_2$, $\theta_0$ and $U_i$ in the lower bound construction}
\end{figure}

\begin{prop}
    Let $1 \leq i \leq k$ and $1 \leq j \leq 2k$. For a vertex $v$ on $U_i$ and a vertex $w$ on $U_j$,
    \begin{align*}
        dist(v, w) &> r\ (j = i + k)\\
        dist(v, w) &\leq r\ (j \neq i + k)
    \end{align*}
\end{prop}
\begin{proof}
    Let $r_v$ and $r_w$ be radial coordinates of $v$ and $w$, respectively. If $j=i+k$, then
    \begin{align*}
        \{ dist(v, w)\}^2 &= r_v^2 + r_w^2-2r_v r_w \cos \angle wo'v\\
        &\geq r_v^2 + r_w^2-2r_v r_w \cos (\pi - \theta_0)\\
        &\geq 2 r_1^2 + 2 r_1^2\cos \theta_0\\
        & > r^2
    \end{align*}
    Otherwise, $\{ dist(v, w)\}^2 \leq 2 r_2^2 + 2 r_2^2\cos 2\theta_0 \leq r^2$
\end{proof}\par
For $1 \leq i \leq k$, let $t_i$ be 0-1 random variables which are equal to 1 if and only if both $U_i$ and $U_{i+k}$ have at least one vertex on themselves.
Let $t=\sum_{i=1}^k t_i$. Then, there exists $O_{t}$ as a vertex-induced subgraph. We are left to lower bound $t$.
\begin{prop}
    \label{F}
    $F(U_1)=\Omega(\theta_0^3)$ as $\theta_0 \rightarrow 0$
\end{prop}
\begin{proof}
    \begin{align*}
        F(U_1)&=\frac{1}{2}\left(r_2^2 \theta_0-r_1^2\theta_0\right)\\
        &=\frac{r^2\theta_0}{4}\left(\frac{1}{1+\cos(2 \theta_0)}
                        - \frac{1}{1+\cos(\theta_0/2)}\right)\\
        &=\frac{r^2\theta_0}{4} \left(\frac{1}{2}+\frac{\theta_0^2}{2}+O(\theta_0^4)
        -\left(\frac{1}{2}+\frac{\theta_0^2}{32}+O(\theta_0^4) \right) \right) \\
        &=\frac{15}{128} r^2 \theta_0^3 +O(\theta_0^5)\\
        &=\Omega(\theta_0^{3})
    \end{align*}
On the third line, we used the Taylor expansion.
\end{proof}
\begin{prop}
    \label{propcn}
    If $k= n^{1/3}$, then there exists a positive constant $C$ such that $\Pr[t \leq C n^{1/3}] \rightarrow 0$ as $n \rightarrow \infty$
\end{prop}
\begin{proof}
    Consider the Poissonization $G_{N_n, r}$. For each $i$, the number of vertices on $U_i$ follows an independent and identical Poisson distribution with mean $n F(U_1)$. Therefore, $t_1,...,t_k$ are independent and identical random variables. Let $p:=1-e^{-n F(U_1)}$. This is equal to the probability that $U_i$ has at least one vertex on itself. By the Proposition \ref{F}, $p=\Omega(1)$. Also, we have $\mathbb{E}[t]=kp^2=\Omega(n^{1/3})$. Since $t_1,...,t_n$ are independent, we can apply the Chernoff bound and obtain 
    $\Pr[t \leq kp^2/2]\leq \exp(-\Omega(n^{1/3}))$. Now we are back to the normal model $G_{n,r}$. Using (\ref{poisson}), it follows that $\Pr[t \leq kp^2/2] \leq \Omega(n^{1/2}) \cdot \exp(-\Omega(n^{1/3}))$. This goes to 0 as $n \rightarrow \infty$.
\end{proof}\par
By combining the proposition with the Theorem \ref{maximallower}, we obtain the main theorem regarding $\mathcal{M}(G_{n,r})$.
\begin{cor}
    There exists a positive constant $C$ such that $\Pr[\exp(Cn^{1/3}) \leq \mathcal{M}(G_{n,r})] \rightarrow 1$
    as $n \rightarrow \infty$.
\end{cor}
\subsection{Upper Bound of the Number of Maximal Cliques}
Let $\overline{vw}$ denote the segment between vertices $v$ and $w$ on a Euclidean plane.
Let $\mathcal{S}:=\{\overline{vw}:v,w\in V(G), dist(v,w)>r\}$. Note that $\mathcal{S}$ is 
like a complement of the random graph.
Two segments $\overline{v_1 v_2}$ and $\overline{w_1 w_2}$ in $\mathcal{S}$ are called
independent if
$\overline{v_1 w_1}$, $\overline{v_1 w_2}$, 
$\overline{v_2 w_1}$, and $\overline{v_2 w_2}$
are not in $\mathcal{S}$.
Two or more segments in $\mathcal{S}$ are
called independent if any pair of the segments are independent.
From the definition, it is obvious that the following proposition holds.
\begin{prop}
There is no $O_{t+1}$ as a vertex-induced subgraph 
$\Leftrightarrow$
There is no set of independent segments $S\subseteq \mathcal{S}$
whose cardinality is $t+1$. 
\end{prop}
Therefore, we are left to bound such $t$ on $G_{n,r}$. With elementary geometry, we can confirm the following.
\begin{prop}
    \label{EI}
    On a Euclidean plane, two independent segments intersect. \\
\end{prop}
We defer the proof to the next section. Thanks to the proposition, we can define an angle between two independent segments thanks to the proposition.

Let $s:=\overline{v_1 v_2}$ and $s':=\overline{w_1 w_2}$ be independent segments. Suppose that the counter-clockwise order of
four endpoints is $v_1w_1w_2v_2$. Let $q$ be the intersection of two segments.
Define a directed angle $\angle(s,s'):=\angle w_1 q v_1$. For convinience, define $\angle(s,s):=0$.
It holds that $\angle(s,s')=\pi - \angle(s',s)$. Therefore, the order of the two segments matters.\par
From here, we consider conditions that two independent segments mWe angle between them is known.
Again, $s=\overline{v_1 v_2},s'=\overline{w_1 w_2} \in \mathcal{S}$ be independent segments. Suppose $v_1w_1v_2w_2$ is the counter-clockwise order of four endpoints.
    Let $m$ be the midpoint of $s$.
    Consider taking a polar coordinate system whose origin and polar axis are $m$ and $\overrightarrow{mv_1}$.
    Let $(r_0,\phi)$ be the polar coordinate of $w_1$. Note that $r_0=dist(m,w_1)$ and $\phi=\angle w_1 m v_1$ (refer to Figure \ref{definitionr}). The next proposition states that if the directed angle $\angle(s,s')$ is small, then $r_0$ is bounded tight.

\begin{prop}
    \label{restrictedr}
    Let $0 \leq \theta_0 \leq \pi/3$. If $\angle(s,s') \leq \theta_0$, then $r_1 \leq r_0 \leq r_2$ where
    \begin{align*}
        r_1&:=(-1/2+\cos \theta_0) \sqrt{\cos \theta_0} \cdot r\\
        r_2&:=(3/2-\cos \theta_0) r
    \end{align*}
\end{prop}

\begin{proof}
    Let $\theta := \angle(s,s')$.
    Let $q$ be the intersection of $s$ and $s'$.
    Let $a := dist(v_1,v_2), b:= dist(w_1,w_2), c = dist(v_1,q), d = dist(w_1,q)$.
    Note that $a$ is the length of $s$ and $a > r$ holds. The same thing stands for $b$.
    From the definition of independence of segments,
    \begin{align*}
        dist(v_1, w_2)&\leq r 
        \Rightarrow c^2+(b-d)^2-2c(b-d) \cos(\pi -\theta) \leq r^2\\
        dist(v_2, w_1)&\leq r
        \Rightarrow (a-c)^2+d^2-2(a-c)d \cos(\pi -\theta) \leq r^2
    \end{align*}
    must hold. From the first inequality, we get
    \begin{align*}
        &c^2+(b-d)^2 - 2c(b-d) \cos(\pi-\theta)\leq r^2\\
        &\Leftrightarrow (c + (b-d) \cos(\theta))^2+((b-d) \sin(\theta))^2\leq r^2\\
        &\Rightarrow c +(b-d) \cos(\theta) \leq r\\
        &\Rightarrow c - d\leq r- b \cos(\theta)
    \end{align*}
    With the same argument, the second inequality yields $- (r - a \cos(\theta)) \leq c - d$.
    Regardless of how $v_1, m, q, v_2$ are lined up on $s$, we have $r_0^2 = (a/2-c)^2+d^2+2(a/2-c)d \cos(\theta)$.
    We can upper and lower bound $r_0$ as below.
    \begin{align*}
        r_0^2 &=(a/2-c)^2+d^2+2(a/2-c)d \cos(\theta) \\
        &\leq (a/2-c+d)^2\\
        &\leq (a/2+r-a\cos(\theta))^2\\
        &\leq r^2(3/2 - \cos(\theta_0))^2
    \end{align*}
    \begin{align*}
        r_0^2 &= (a/2-c)^2+d^2+2(a/2-c)d \cos(\theta) \\
        &\geq (a/2-c+d)^2 \cos(\theta)\\
        &\geq (a/2-r+b \cos(\theta))^2\cos(\theta)\\
        &\geq r^2(-1/2 + \cos(\theta_0))^2\cos(\theta_0)
    \end{align*}
\end{proof}\\
\begin{figure}[H]
\centering
\begin{center}
    \includegraphics[width = 12cm]{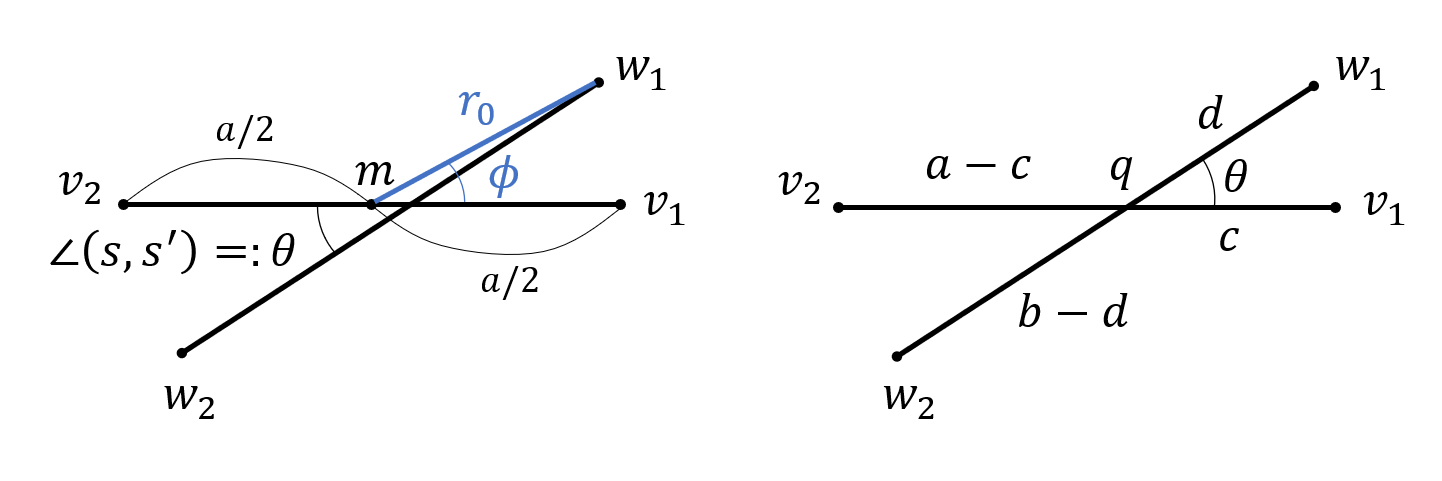}
\end{center}
\caption{Definitions of points, length, and angles of two independent segments in the Proposition \ref{restrictedr}}
\label{definitionr}
\end{figure}

If $v_1,q,m,v_2$ are lined up in this order, we can bound $\phi$ as $0 \leq \phi \leq \theta \leq \theta_0$. By considering the other case ($v_1,m,q,v_2$ is lined up in this order)
in the same way, we have the following.
\begin{cor}
    \label{regionU}
    Let $0 \leq \theta_0 \leq \pi/3$. If $\angle(s,s') \leq \theta_0$, 
    then either $w_1$ or $w_2$ must lie on a region $U_1$ or a region $U_2$ where
    \begin{align*}
        U_1&:=\{(r_0,\phi): r_1 \leq r_0 \leq r_2, 0 \leq \phi \leq \theta_0 \}\\
        U_2&:=\{(r_0,\phi): r_1 \leq r_0 \leq r_2, \pi \leq \phi \leq \pi+\theta_0 \}\\
    \end{align*}
\end{cor}

\begin{figure}[H]
\centering
\begin{center}
    \includegraphics[width = 8.5cm]{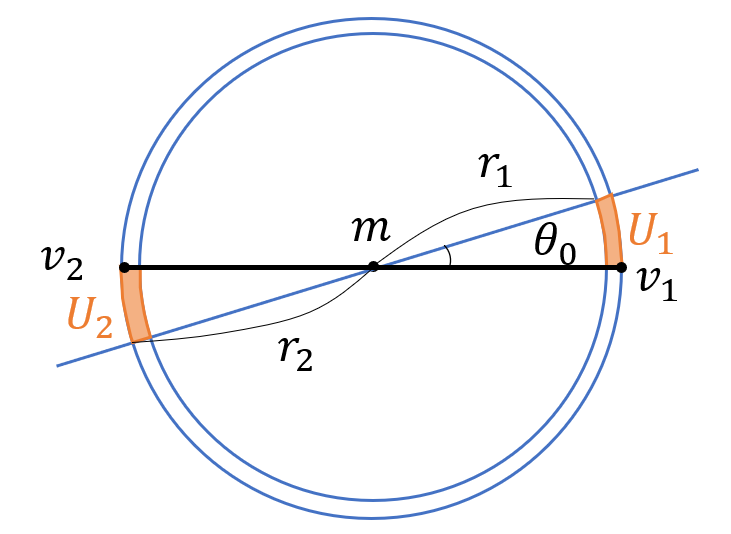}
\end{center}
\caption{$U_1$ and $U_2$ in the Corollary \ref{regionU}}
\end{figure}

Then we obtain the next lemma: if the directed angle is bounded small, the expected number of independent segments is also small.
\begin{lem}
    \label{U}
    Let $s \in \mathcal{S}$ be a segment.
    For a constant $0 \leq \theta_0 \leq \pi/3$,
    let $X$ be a number of segments $s'$ where $s$ and $s'$ are independent and
    $\angle(s,s')\leq \theta_0$. Then $\mathbb{E}[X] = n O(\theta_0^3)$ as $\theta_0 \rightarrow 0$.
\end{lem}
\begin{proof}
    By the Corollary \ref{regionU}, $X$ is at most the number of vertices in $U:=U_1 \cup U_2$. Therefore, $\mathbb{E}[X]\leq n F(U)$.
    We are left to prove $F(U_1)=F(U_2)=O(\theta_0^3)$ as $\theta_0 \rightarrow 0$, and
    \begin{align*}
        F(U_1) &\leq \frac{1}{2} r_2^2 \theta_0 - \frac{1}{2} r_1^2 \theta_0 \\
        &=\frac{1}{2}r^2 \theta_0 \left( (3/2 - \cos(\theta_0))^2
        - (-1/2+\cos(\theta_0))^2 \cos(\theta_0) \right)\\
        &=\frac{1}{2}r^2 \theta_0 \left( (2-2 \cos \theta_0) +
        (-1/2+\cos \theta_0)^2(1-\cos \theta_0)\right)\\
        &= \frac{9}{16}r^2 \theta_0^3 + O(\theta_0^5)
    \end{align*}
    To derive the fourth line, we use the Taylor expansion.
\end{proof} \par

To fully use the Lemma \ref{U}, we prove another lemma, which states that if there exists large $O_t$, we can always take a set of segments so that its size is unneglectable, and the directed angles between them are small.
\begin{lem}
    \label{pigeon}
    Let $k$ be a positive integer.
    If there exists a set of $t$ independent segments $S\subseteq \mathcal{S}$, then there exists a segment $s'$ 
    and a set of independent segments $S' \subseteq \mathcal{S}$ such that
    \begin{align}
        \label{three}
        \left\{
        \begin{array}{l}
        |S'| \geq \left \lceil t/k \right \rceil\\
        s' \in S'\\
        \forall s''\in S'. \angle(s', s'') \leq \pi/k\\
        \end{array}
        \right.
    \end{align}
\end{lem}
\begin{proof}
    Let $s_0=\overline{v_1v_2}$ be some segment in $S$. For an integer $1\leq i \leq k$, define \\$S_i := \left\{ s\in S: 
        \frac{i-1}{k} \pi \leq \angle(s_0,s) < \frac{i}{k} \pi \right\}$.
    Since $\bigcup_{1\leq i \leq k}S_i=S$, we have
    $\sum_{i=1}^{k} |S_i|=t$. By the pigeonhole principle, there exists an index $i$
    such that $|S_i|\geq \lceil t/k \rceil$. Consider taking the segment $s'$ in $S_i$ so that
        $\angle(s_0,s')=\min_{s''\in S_i} \angle(s_0,s'')$.
    If we can prove $\angle(s',s'')\leq \angle(s'',s_0)-\angle(s',s_0)$, we are done.
    Let $a$ be the intersection of $s_0$ and $s'$, $b$ be that of $s_0$ and $s''$,
    and $c$ be that of $s'$ and $s''$. If $a = b$, then it obviously holds.
    We have two other cases: $v_1,b,a,v_2$ are lined up on $s_0$ in this order, and $v_1,a,b,v_2$ are lined up on $s_0$ in this order.
    For the former case, $\angle a=\angle(s_0,s')$, $\angle b=\pi - \angle(s_0,s'')$ and
    $\angle c = \angle(s',s'')$ where $\angle a$, $\angle b$, $\angle c$ are the inner angle of $\triangle abc$. Therefore,
    \begin{gather}
        \angle a + \angle b + \angle c = \pi \label{triangleangles}\\
        \Leftrightarrow \angle(s',s'') = \pi - (\angle(s_0,s')) - (\pi-\angle(s_0,s''))=\angle(s_0,s'')-\angle(s_0,s') \notag
    \end{gather}
    Here, (\ref{triangleangles}) follows from the fact that the sum of the inner angles of a triangle is equal to $\pi$.\\
    The latter case can be shown similarly.
\end{proof}\par
We finally combine the two lemmas to upper bound $\tau$. After that,
we apply Theorem \ref{maximalupper} to get the upper bound of the number of maximal cliques.

\begin{thm}
    \label{EOt}
    Let $\epsilon$ be a positive constant.
    Let $t=n^{1/3+\epsilon}$. Then,
    \begin{align*}
        \Pr[G_{n,r} \text{ has $O_t$ as its vertex-induced subgraph}] \rightarrow 0
    \end{align*}
    as $n \rightarrow \infty$.
\end{thm}
\begin{proof}
Let $k:=n^{1/3}$ and $\theta_0:=\pi/k$.
    Suppose $G_{n,r}$ has $O_t$ as its vertex-induced subgraph.
    Apply Lemma \ref{pigeon} for $t$ and $k$. Then there exists a segment $s'$ and a set of independent segments $S'$ such that
    (\ref{three}) holds.
    We claim that for a certain segment $s'$, the probability that
    there exists a set of independent segments $S'$ such that (\ref{three}) holds is at most $\exp(-\Omega(n^{\epsilon}))$.
    By Lemma \ref{U}, the number of segments that satisfy the third condition of (\ref{three}) is at most the number of vertices on $U$, and that expected number is bounded by $O(n \theta_0^3)=O(1)$.
    However, to satisfy $(\ref{three})$, The cardinarity of $S'$ needs be at least $n^{\epsilon}$. This means that the number of vertices on $U$ also needs to be at least $n^{\epsilon}$. 
    Since the vertices are placed independently, we can apply the Chernoff bound to get the probability bound $\exp(-\Omega(n^{\epsilon}))$. \par
    Finally, we apply a union bound over all segments of which there are at most $n(n-1)/2$. We have
    \begin{align*}
        \Pr[G_{n,r} \text{ has $O_t$ as its vertex-induced subgraph}] \leq \frac{n(n-1)}{2} \exp(-\Omega(n^{\epsilon}))
    \end{align*}
    This goes to 0 as $n \rightarrow \infty$.
\end{proof}
\begin{cor}
    There exist a positive constant $C$ such that for all $\epsilon > 0$,
        $\Pr[\mathcal{M}(G_{n,r})\leq \exp(Cn^{1/3+\epsilon})]\rightarrow 1$
    as $n\rightarrow \infty$
\end{cor}

\section{Hyperbolic Random Graphs}
\subsection{Definitions}
Given $n \in \mathbb{N^+}$, $\gamma \in (2, \infty)$, and $C \in \mathbb{R}$, let $R := 2 \log n + C$, and $\alpha = (\gamma-1)/2$.
 A hyperbolic random graph $G_{n,\gamma, C}$ is obtained
as below:
\begin{itemize}
    \item The vertex set is $V =\{1,2,...,n\}$.
    \item The vertices are identically and independently distributed on a hyperbolic plane. 
    The probability density function by a polar coordinate is:
    \begin{align*}
        f(r,\theta)=
        \begin{cases}
        \frac{1}{2\pi}\cdot \frac{\alpha \sinh (\alpha r)}{\cosh(\alpha R)-1} & (r \leq R)\\
        0 & (r > R)
        \end{cases}
    \end{align*}
    \item The edge set $E$ is given by $\{(u,v):dist(u,v)\leq R\}$
\end{itemize}

\begin{figure}[H]
\centering
\begin{center}
    \includegraphics[width = 8cm]{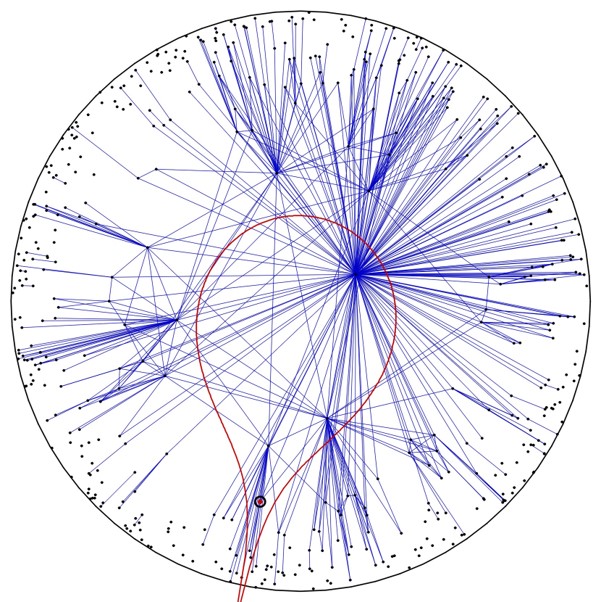}
\end{center}
\caption{A hyperbolic random graph ($n=500$, $\gamma=2.1$) }
\end{figure}

In our work, we are interested in the case $2 < \gamma < 3\ (1/2 < \alpha < 1)$ where the generated graph is ``scale-free" with high probability.
Let $(r_u,\phi_u)$ and $(r_v,\phi_v)$ be polar coordinates of 
$u$ and $v$ respectively.
Let $\theta:=\pi-|\pi - |\phi_u-\phi_v||$. Then, the hyperbolic cosine formula suggests that
\begin{align*}
    \cosh(dist(u,v))=\cosh r_u \cosh r_v - \sinh r_u \sinh r_v \cos \theta
\end{align*}
Again, define $F(U):=\int_U f(r,\theta)dr d\theta$. 
On a hyperbolic plane, the area of $U$ is equal to
$\int \sinh r dr d\theta=: \mu(U)$. Let $\rho(r,\theta):= \frac{dF}{d\mu} = \frac{f(r,\theta)}{\sinh r}$.
Since the value of $\rho(r, \theta)$ does not depend on $\theta$, we sometimes denote it as $\rho(r)$.
Intuitively, $\rho$ is the
the probability that a vertex lies on a single unit square around $(r, \theta)$. By differentiating $\rho$, we can confirm that it is monotonically decreasing with respect to $r$. Thus, the graph is dense near the origin of the plane. \par
Let $o$ be the origin of the plane. We define $B_0(d)$ as the set of points $p$ such that $dist(p,o)<d$. In other words, $B_0(d):=\{(r,\phi): r < d\}$.

\subsection{Lower Bound of the Number of Maximal Cliques}
To obtain the lower bound of $\mathcal{M}(G_{n,r,C})$, 
we do the same thing as the Euclidean random geometric graphs i.e. We take regions so that vertices on them form $O_t$. \par
Let $k \geq 4$ be an integer.
Let $\theta_0:=\pi/(3k)$.
For $1\leq i \leq 2k$, Define $U_i:=\{(r,\phi): r_1 \leq r \leq r_2,\ 3(i-1)\theta_0 \leq \phi \leq (3(i-1)+1)\theta_0\}$ where $r_1$ and $r_2$ are determined as
\begin{align*}
    \cosh(2 r_1)&=\frac{\cosh(R)}{\cos^2(\theta_0/2)}\\
    \cosh(2 r_2)&=\frac{\cosh(R)-1}{\cos^2(\theta_0)}
\end{align*}
We now assume that $r_1<r_2 \leq R$. We later confirm this holds under proper $\theta_0$ and $n$.
\begin{prop}
    Let $1 \leq i \leq k$ and $1 \leq j \leq 2k$. For a vertex $v$ on $U_i$ and a vertex $w$ on $U_j$,
    \begin{align*}
        dist(v, w) &> R\ (j = i + k)\\
        dist(v, w) &\leq R\ (j \neq i + k)
    \end{align*}
\end{prop}
\begin{proof}
    Let $r_v$ and $r_w$ be radial coordinates of $v$ and $w$, respectively. If $j=i+k$, then
    \begin{align*}
        \cosh (dist(v,w))&= \cosh r_v \cosh r_w - \sinh r_v \sinh r_w \cos \angle wov \\
        &\geq \cosh r_v \cosh r_w - \sinh r_v \sinh r_w \cos (\pi - \theta_0)\\
        &\geq \cosh^2 r_1 + \sinh^2 r_1 \cos \theta_0\\
        &\geq \cos^2 (\theta_0/2) \cosh(2 r_1)+\sin^2(\theta_0/2)\\
        & >\cosh(R)
    \end{align*}
    Otherwise, 
    \begin{align*}
        \cosh (dist(v,w)) &\leq \cos^2 (\theta_0)\cosh(2 r_2) + \sin^2(\theta_0)\\
        &\leq \cosh(R) - 1 + \sin^2(\theta_0) \leq \cosh(R)
    \end{align*}
\end{proof}\par
For $1 \leq i \leq k$, let $t_i$ be 0-1 random variables which are equal to 1 if and only if both $U_i$ and $U_{i+k}$ have at least one vertex on themselves.
Let $t=\sum_{i=1}^k X_i$. Then, there exists $O_{t}$ as a vertex-induced subgraph. Again, we are left to lower bound $t$.
\begin{prop}
    \label{HF}
    If $\theta_0 > 1/\sqrt{n}$, then
    $F(U_1)=\Omega(n^{-\alpha}\theta_0^3)$ as $n \rightarrow \infty$ and $\theta_0 \rightarrow 0$
\end{prop}
\begin{proof}
    We first claim that $r_2=\log n + \Theta(1)$. It is sufficient to prove $\cosh(2r_2)=\Theta(n^2)$, and this follows since $R=2 \log n + \Theta(1)$ and $\cos(\theta_0) > \frac{1}{2}$. With the same argument, we have $r_1=\log n + \Theta(1)$. Then, $F(U_1)$ can be calculated as below.
    \begin{align*}
        F(U_1)&= \int_{r_1}^{r_2} \int_0^{\theta_0} f(r, \theta) dr d\theta \\
        &= \int_{r_1}^{r_2} \int_0^{\theta_0} \rho(r, \theta) \sinh(r) dr d \theta \\
        &\geq \rho(r_2) \int_{r_1}^{r_2} \int_0^{\theta_0}\sinh(r) dr \\
        &= \frac{\theta_0}{2\pi} \frac{\alpha}{\cosh(\alpha R)-1}
        \frac{\sinh(\alpha r_2)}{\sinh(r_2)} (\cosh(r_2) - \cosh(r_1))\\
        &= \Theta(\theta_0 n^{-1-\alpha})(\cosh(r_2) - \cosh(r_1)) \\
    \end{align*}
    On the third line, we use the fact that $\rho(r)$ is a monotonically decreasing function.
    Here,
    \begin{align*}
        &\cosh(r_2) - \cosh(r_1)\\
        &=\frac{1}{\cosh(r_2)+\cosh(r_1)}(\cosh^2(r_2)-
        \cosh^2(r_1))\\
        &=\Theta(n^{-1})(\cosh(2r_2)-\cosh(2r_1))\\
        &=\Theta(n^{-1})\cosh(R)
        \left(\frac{1}{\cos^2(\theta_0)}-\frac{1}{\cos^2(\theta_0/2)}
        +\frac{1}{\cosh(R)\cos^2(\theta_0/2)} \right)\\
        &=\Theta(n)
        \left(\frac{3}{4}\theta_0^2+O(\theta_0^4)+\Theta(n^{-2}) \right)\\
        &=\Theta(n \theta_0^2)
    \end{align*}
    The last equality holds since $\theta_0^2 > n^{-1} \gg n^{-2}$.
\end{proof}\par
As a byproduct of the Proposition \ref{HF}, we obtain the condition which $r_1<r_2\leq R$ holds.
\begin{cor}
    If $\theta_0 > 1/\sqrt{n}$, then with sufficiently large $n$ and sufficiently small $\theta_0$, $r_1<r_2\leq R$.
\end{cor}
\begin{proof}
    Recall that $r_2=\log n + \Theta(1)$ and $R=2 \log n + \Theta(1)$. Therefore, we have $r_2 \leq R$ with sufficiently large $n$. Also, recall that $\cosh(r_2) - \cosh(r_1) = \Theta(n \theta_0^2) > 0$ as $n \rightarrow \infty$ and $\theta_0 \rightarrow 0$, proving $r_1<r_2$ with sufficiently large $n$ and sufficiently small $\theta_0$.
\end{proof}
\begin{prop}
    If $k= n^{(1-\alpha)/3}$, then there exists a positive constant $C$ such that $\Pr[t \leq C n^{(1-\alpha)/3}] \rightarrow 0$ as $n \rightarrow \infty$
\end{prop}
\begin{proof}
    As in the Proposition \ref{propcn}, consider the Poissonization $G_{N_n,\gamma, C}$. The proof is the same. We only need to comfirm that $\theta_0=\pi/(3n^{(1-\alpha)/3}) > 1/\sqrt{n}$ for sufficiently large $n$ and $n F(U_1)=\Omega(1)$.
\end{proof}

Again, by combining the proposition with the Theorem \ref{maximallower}, we obtain the main theorem regarding $\mathcal{M}(G_{n,\gamma,C})$.

\begin{cor}
    There exists a positive constant $C'$ such that $\Pr[\exp(C'n^{(1-\alpha)/3}) \leq \mathcal{M}(G_{n,\gamma,C})] \rightarrow 1$
    as $n \rightarrow \infty$.
\end{cor}

\subsection{Upper Bound of the Number of Maximal Cliques}
To upper bound $\mathcal{M}(G_{n, \gamma, C})$, we prove the hyperbolic versions of the Lemma \ref{U} and \ref{pigeon}. Again, we start with defining segments.
Let $\overline{vw}$ denote the segment (i.e., the shortest geodesic) between vertices $v$ and $w$ on a hyperbolic plane.
Let $\mathcal{S}:=\{\overline{vw}:v,w\in G(V), dist(v,w)>R\}$. We define the independence of
segments in the same way as for the Euclidean case. We first confirm that two independent segments intersect. To do so, we use the following lemma from \cite{hyperbolicclique}.

\begin{lem}[\cite{hyperbolicclique}]
    Let $a$ and $b$ be points on a hyperbolic plane where $dist(a,b)=R$. Let $l$ be a line that passes through $a$ and $b$. Let $c$ and $d$ be points on the same halfplane induced by $l$. Suppose $dist(a,c) \leq R$, $dist(b,c) \leq R$, $dist(a,d)\leq R$, and $dist(b,d) \leq R$. Then, $dist(c,d)\leq R$.
\end{lem}
Also, we use the following basic facts of non-Euclidean geometry.
\begin{prop} [Euclid's Proposition I-13]
    \label{euclid1}
    If a line stood on another line makes angles, it will make
    two right angles, or the sum of the angles is equal to two right angles.
\end{prop}

\begin{prop} [Euclid's Proposition I-19]
    \label{euclid3}
     In any given triangle, the side for a larger angle is larger.
\end{prop}

\begin{prop} [Saccheri–Legendre Theorem \cite{triangle}]
    \label{euclid2}
    The sum of all three inner angles in a triangle is less than or equal to two right angles.
\end{prop}

Note that all of those Lemma and Propositions do not rely on the parallel postulate. Therefore, they are valid on both a Euclidean and a hyperbolic plane.\\
\begin{prop}
    \label{HI}
    On a hyperbolic plane, Two independent segments intersect.
\end{prop}
\begin{proof}
    Let $s:=\overline{v_1 v_2}$ and $s':=\overline{w_1 w_2}$ be two independent segments. Let $l$ be a line that passes through $v_1$ and $v_2$. We claim that $w_1$ and $w_2$ are located on the different halfplanes induced by $l$. We prove this by contradiction. Suppose $w_1$ and $w_2$ are on the same halfplane. Take points $p$ and $q$ on the segment $s$ so that $dist(v_1, p)=(a - R) / 2$, $dist(p,q)=R$, and $dist(q,v_2)=(a-R)/2$, where $a:=dist(v_1,v_2)$. If we can prove $dist(p,w_1) \leq R$, then we also have $dist(q,w_1) \leq R$, $dist(p,w_2) \leq R$, and $dist(q,w_2) \leq R$ by symmetricity. By Lemma 3.5, it follows that $dist(w_1,w_2) \leq R$. However, this is a contradiction since $s' \in \mathcal{S}$. \par 
    Now the goal is to prove $dist(p,w_1) \leq R.$ Focus on $\triangle v_1 w_1 v_2$. By the definition of independent segments, $\overline{w_1 v_2}$ is shorter than $\overline{v_1 v_{2}}$. Therefore, by the Proposition \ref{euclid2}, $\angle w_1 v_1 v_2 < \angle v_1 w_1 v_2$. We have $\angle w_1 v_1 v_2 < \pi/2$. Otherwise, $\angle w_1v_1v_2 \geq \pi/2$ and $\angle v_1w_1 v_2 > \pi/2$, and these contradict with the Proposition \ref{euclid3}. With the same argument, $\angle w_1 p v_2 < \pi/2$. By the Proposition \ref{euclid1}, $w_1 p v_1 = \pi - \angle w_1 p v_2 \geq \pi/2$. Therefore, on $\triangle w_1 p v_1$, $\angle w_1 p v_1 > \angle w_1 v_1 p$. Again, by the Proposition \ref{euclid3}, the length of $\overline{v_1 w_1}$ is greater than $\overline{p w_1}$. Recall that $dist(v_1, w_1) \leq R$. Thus, We can conclude that $dist(p, w_1) \leq R$.
\end{proof}\par
Notice that the proof of the Proposition \ref{HI} does not rely on the parallel postulate. Therefore, this proof is also valid on a Euclidean plane, proving the Proposition \ref{EI}. \par
We can also define the angle of segments on a hyperbolic plane. Again, $s=\overline{v_1 v_2},s'=\overline{w_1 w_2} \in \mathcal{S}$ be independent segments. Suppose $v_1w_1v_2w_2$ is the counter-clockwise order of four endpoints.
    Let $q$ be the intersection of the two segments. Define $\angle(s,s'):=\angle w_1qv_1$. Let $m$ be the midpoint of $s$.
    Consider taking a polar coordinate system whose origin and polar axis are $m$ and $\overrightarrow{mv_1}$. 
    Let $(r_0,\phi)$ be the polar coordinate of $w_1$. Also, define $a := dist(v_1,v_2), b:= dist(w_1,w_2), c = dist(v_1,q), d = dist(w_1,q)$. 
    The definitions are summarized in Figure \ref{definitionr}.

\begin{lem}
    \label{hypersegments}
    Let $\theta := \angle(s,s')$. For all $0 < \theta < \pi$,
    \begin{gather}
        R < a, b\leq R + 2 \Delta(\theta) \label{lemsegment1}\\
        e^c,e^d \geq \frac{\sin \theta}{2} \frac{1}{1+e^{-2R}}e^{R/2} \label{lemsegment2}\\
        \cosh(R/2- \Delta(\theta)) \frac{1+\cos(\theta)}{2}\leq \cosh(r_0)\leq \cosh(R/2+2\Delta(\theta))
        \label{lemsegment3}
    \end{gather}
    where $\Delta(\theta):=\log \left( 2(1+e^{-2 R})/(1+\cos(\theta))\right)$
\end{lem}
\begin{proof}
    By the definition of independent segments, 
    $dist(v_1, w_2)$,
    $dist(v_2, w_1)$,
    $dist(v_1, w_1)$, and
    $dist(v_2, w_2)$ are all less than or equal to $R$. Therefore,
    \begin{align}
        \cosh(c) \cosh(b - d) + \sinh(c) \sinh(b - d) \cos(\theta) &\leq \cosh(R) \label{hcf1} \\
        \cosh(a-c) \cosh(d) + \sinh(a-c) \sinh(d) \cos(\theta) &\leq \cosh(R) \label{hcf2}\\
        \cosh(c) \cosh(d) - \sinh(c) \sinh(d) \cos(\theta) &\leq \cosh(R) \label{hcf3}\\
        \cosh(a - c) \cosh(b - d) - \sinh(a - c) \sinh(b - d) \cos(\theta) 
        &\leq \cosh(R) \label{hcf4}
    \end{align}
    From (\ref{hcf1}), we have
    \begin{align}
        (\ref{hcf1}) \Leftrightarrow
        \cosh(b-d+c) &\frac{1+\cos(\theta)}{2} + 
        \cosh(b-d-c) \frac{1-\cos(\theta)}{2} \leq \cosh(R) \notag \\
        \Rightarrow &\frac{1+\cos(\theta)}{2} e^{b-d+c}\leq e^R + e^{-R} \notag \\
        \Rightarrow &b-d+c \leq R + \Delta(\theta) \label{delta1}
    \end{align}
    By doing the same thing on (\ref{hcf2}), we obtain
    \begin{align}
        -(R-a)-\Delta(\theta)&\leq c-d \leq R-b+\Delta(\theta) \notag \\
        \Rightarrow -\Delta(\theta) &\leq c-d \leq \Delta(\theta)
        \label{delta2}
    \end{align}
    The last inequalities hold since $a,b > R$. Plugging this in to (\ref{delta2}) yields (\ref{lemsegment1}).\\
    From (\ref{hcf4}), we have
    \begin{align*}
        (4.6)\Leftrightarrow
        \cosh(a-c+b-d) &\frac{1-\cos(\theta)}{2}+\cosh(a-c-b+d)\frac{1+\cos(\theta)}{2}
        \leq \frac{e^R+e^{-R}}{2}\\
        \Rightarrow &\frac{1-\cos(\theta)}{2} e^{2R-c-d}\leq e^R + e^{-R}\\
        \Rightarrow &\frac{1-\cos(\theta)}{2} \frac{e^{R}}{1+e^{-2R}} \leq e^{c+d}
    \end{align*}
    Multiplying this with (\ref{delta2}) gives us
    \begin{align*}
        e^{c}, e^{d} \geq \frac{\sin \theta}{2} \frac{1}{1+e^{-2R}} e^{R/2}
    \end{align*}
    For $r_0$,
    \begin{align*}
        \cosh(a/2-c) \cosh(d) + \sinh(a/2-c) \sinh(d) \cos(\theta) = \cosh(r)\\
    \end{align*}
    Note that this holds no matter how $m$ and $q$ are  lined up on $s$. We have
    \begin{align*}
        \cosh(a/2-c+d)\frac{1+\cos(\theta)}{2} \leq \cosh(r)\leq \cosh(a/2-c+d)
    \end{align*}
    Finally, use (\ref{delta2}) to derive (\ref{lemsegment3}).
\end{proof}

\begin{cor}
    \label{hyperregionU}
    Let $\theta_0 \in (0, \pi)$. If $\angle(s,s') \leq \theta_0$, then either $w_1$ or $w_2$ must lie on $U_1$ or $U_2$ where
    \begin{align*}
        U_1&:=\{(r_0,\phi): r_1 \leq r_0 \leq r_2, 0 \leq \phi \leq \theta_0 \}\\
        U_2&:=\{(r_0,\phi): r_1 \leq r_0 \leq r_2, \pi \leq \phi \leq \pi+\theta_0 \}\\
    \end{align*} Here,
    \begin{align*}
        r_1&:=\cosh^{-1} \left(\cosh(R/2- \Delta(\theta_0)) \frac{1+\cos(\theta_0)}{2} \right)\\
        r_2&:=R/2+2 \Delta(\theta_0)
    \end{align*}
\end{cor}

\begin{cor}
    \label{hyperU}
    Let $s$ be a segment.
    Let $\theta_0 > 1/\sqrt{n}$.
    Let $X$ be a number of segments $s'$ which is independent from $s$ and
    $\angle(s,s')\leq \theta_0$, then
    \begin{align*}
        \mathbb{E}[X] = \sup_{(r,\theta)\in U} \{\rho(r,\theta) \}O(n^2\theta_0^3)
    \end{align*}
    as $n \rightarrow \infty$ and $\theta_0 \rightarrow 0$.
\end{cor}
\begin{proof}
    In this proof, we simply denote $\Delta(\theta_0)$ as $\Delta$.
    It is sufficient to prove that \\$F(U_1)=\sup_{(r,\theta)\in U_1} \{\rho(r,\theta) \}O(n \theta_0^3)$. Here,
    \begin{align*}
        F(U_1)
        &=\int_{U_1} f(r,\theta) dr d\theta\\
        &=\int_{U_1} \rho(r,\theta) d\mu\\
        &\leq\sup_{U_1}\rho(r,\theta)\int_{U_1} d\mu\\
        &=2 \theta_0\sup_{U_1}\rho(r,\theta) (\cosh(r_2) - \cosh(r_1))\\
        &=2 \theta_0\sup_{U}\rho(r,\theta) (\cosh(R/2+2 \Delta) - \cosh(R/2 - \Delta) 
        \\&+ \frac{1-\cos(\theta_0)}{2} \cosh(R/2-\Delta) )
    \end{align*}
    Here, we can bound
    \begin{align*}
        &\cosh(R/2+2\Delta) - \cosh(R/2-\Delta)\\&=
        \frac{\cosh(R/2+2\Delta) - \cosh(R/2-\Delta)}{3\Delta} \cdot 3\Delta\\
        &\leq 3\sinh(R/2+2 \Delta) \cdot \Delta\\
        &\leq 3\sinh(R/2+2 \Delta) \cdot \left( \frac{2(1+e^{-2R})}{1+\cos \theta_0}-1 \right)\\
        &= O(n) \cdot \left( \frac{1-\cos \theta_0+2e^{-2R}}{1+\cos \theta_0} \right)\\
        &= O(n) \cdot \left(O(\theta_0^2) + O(n^{-4}) \right)=O(n\theta_0^2)
    \end{align*}
    The inequality follows because $\cosh(x)$ is a convex function. The last equality follows by the assumption $\theta_0^2>n^{-1} \gg n^{-4}$.
    Also,
    \begin{align*}
        \frac{1-\cos(\theta_0)}{2}\cosh(R/2-\Delta) =O(\theta_0^2) \cdot O(n)=O(n\theta_0^2)
    \end{align*}
\end{proof}\par
To fully use the Lemma \ref{hyperU}, we prove the hyperbolic version of the Lemma \ref{pigeon}. The statement is identical.
\begin{lem}
    \label{hyperpigeon}
    Let $k$ be a positive integer.
    If there exists a set of $t$ independent segments $S\subseteq \mathcal{S}$, then there exists a segment $s'$ 
    and a set of independent segments $S' \subseteq \mathcal{S}$ such that
    \begin{align}
        \label{hthree}
        \left\{
        \begin{array}{l}
        |S'| \geq \left \lceil t/k \right \rceil\\
        s' \in S'\\
        \forall s''\in S'. \angle(s', s'') \leq \pi/k\\
        \end{array}
        \right.
    \end{align}
\end{lem}
\begin{proof}
The proof of Lemma \ref{pigeon} is also valid on a hyperbolic plane, 
except for the part (\ref{triangleangles}) where we used the fact that the sum of inner angles of a triangle is equal to $\pi$. This can be replaced with the Proposition \ref{euclid3}. Then we get
\begin{gather}
        \angle a + \angle b + \angle c \leq \pi \\
        \Leftrightarrow \angle(s',s'') \leq \pi - (\angle(s_0,s')) - (\pi-\angle(s_0,s''))= \angle(s_0,s'')-\angle(s_0,s') \notag
    \end{gather}
    Thus, we can prove the same statement.
\end{proof}\par
It is not so hard from here to prove the limited version
of the main theorem. The next proposition states that if vertices of $O_t$ do not fall on a region near the origin,
we can bound its size.
\begin{prop}
    Let $V' :=\{v \in V(G_{n,\gamma,C}): \text{the radial coordinate of $v$ is greater than $R/2$}\}$. Let $G_{n,\gamma,C}'$ be a subgraph of $G_{n,\gamma,C}$ induced by $V'$.
    Let $t=n^{(1-\alpha)/3+\epsilon}$ where $\epsilon$ is arbitrary positive constant. Then,
    \begin{align*}
        \Pr[G_{n,\gamma,C}'\ \text{has $O_t$ as its vertex-induced subgraph}] \rightarrow 0
    \end{align*}
    as $n \rightarrow \infty$.
\end{prop}
\begin{proof}
    We will not go into detail here. However, the statement can be proved in the same way as the Theorem \ref{EOt}. We end up claiming $\rho(R/2) \cdot O(n^2 \theta_0^3)=O(1)$ when $\theta_0:=\pi/n^{(1-\alpha)/3}$. Since $\rho(R/2)=O(n^{-1-\alpha})$, it follows.
\end{proof}
\subsection{Dealing With Non-Uniformity}
We prove the additional facts about two independent segments.
Again, we will use the definitions of points, angles, and length in Figure \ref{definitionr}.
\begin{cor}
    For all $0 < \angle(s,s') < \pi$,
        $\frac{1}{2}\cosh(R/2- \Delta(\pi/2)) \leq \cosh(r_0)\leq \cosh(R/2+2\Delta(\pi/2))$.
\end{cor}
\begin{proof}
    Let $\theta:=\angle(s,s')$. Due to symmetricity, ($\ref{lemsegment3}$) in the Lemma \ref{hypersegments} also holds for $\pi - \theta$. Especially, when $\theta = \pi - \theta'$ for some $0 < \theta' \leq \pi/2$, we have
    \begin{align*}
\cosh(R/2- \Delta(\theta')) \frac{1+\cos(\theta')}{2}\leq \cosh(r_0)\leq \cosh(R/2+2\Delta(\theta'))
    \end{align*}
    Since $\cosh(R/2-\Delta(\theta))\frac{1+\cosh(\theta)}{2}$ and $\cosh(R/2+2 \Delta(\theta))$ are monotonic functions on $0 < \theta \leq \pi/2$, we have the bound on $\cosh(r_0)$.
\end{proof}\par
\begin{cor}
    \label{hyperr0}
    With sufficiently large $n$, $\frac{1}{8} \cosh(R/2) < \cosh(r_0)$ and $R/2- 3 \log 2 < r_0$.
\end{cor}
\begin{proof}
    It is sufficient to prove that we have $ \frac{1}{8} \cosh(R/2) < \cosh(R/2-3 \log 2) < \frac{1}{2} \cosh(R/2-\Delta(\pi/2))$ if $n$ is large enough. We omit the detail here since it is a simple calculation task.
\end{proof}\par
The following proposition states that $\phi$ can be lower bounded in terms of $\theta:=\angle(s,s')$.

\begin{prop}
    \label{polar}
    if $\theta > 1/\sqrt{n}$, then
    with sufficiently large $n$,
    \begin{align*}
        \sin \phi \geq \frac{\sin^2 \theta}{9}
    \end{align*}
\end{prop}
\begin{proof}
    If $\theta = \pi/2$, it clearly holds.
    Otherwise, 
    Let $\theta := \angle(s,s')$.
    Let $q$ be the intersection of $s$ and $s'$.
    Let $a := dist(v_1,v_2), b:= dist(w_1,w_2), c = dist(v_1,q), d = dist(w_1,q)$.
    By the sine formula on a hyperbolic plane, it holds (no matter how $m$ and $q$ are lined up on $s$) that 
    \begin{align*}
        \frac{\sin \theta}{\sinh r_0} = \frac{\sin \phi}{\sinh d}
    \end{align*}
    Therefore,
    \begin{align*}
        \sin &\phi \geq  \frac{\sinh d}{\sinh r_0} \sin \theta \\
        &\geq \frac{e^d - e^{-d}}{e^{r_0}} \sin \theta \\
        &\geq \frac{\frac{\sin \theta}{2} \frac{e^{R/2}}{1+e^{-2R}} - 1}{e^{R/2} e^{2 \Delta(\pi/2)}} \sin \theta\\
        &\geq \frac{\sin^2 \theta}{8}\left( \frac{1-e^{-R}}{(1+e^{-2R})^3} - \frac{8}{e^{R/2} e^{2 \Delta(\pi/2)} \sin \theta} \right)\\
        &\geq \frac{\sin^2 \theta}{9}
    \end{align*}
    For sufficiently large $n$ since $e^{-R} \rightarrow 0$ and $e^{R/2} \sin \theta =\Omega(\sqrt{n}) \rightarrow +\infty$
\end{proof}

\begin{lem}
    \label{toobigangle}
    Let $s$ and $s'$ be independent segments.
    Let $h$ be the distance between the midpoint $m$ of $s$ and the origin.
    Suppose one of $s$'s endpoints is in $B_0(R/2)$, and $\sin (\angle(s,s')) \geq 12 e^{-h/4}$, then with sufficiently large $n$,
    both endpoints of $s'$ are outside of $B_0(R/2)$.
\end{lem}
\begin{proof}
    Note that we would never have the case where both endpoints of a segment in $\mathcal{S}$ are in $B_0(R/2)$. We prove the proposition by contradiction. By symmetricity and the fact that $\sin (\pi - \angle(s,s')) = \sin(\angle(s,s'))$, We can assume that $w_1$ is in $B_0(R/2)$ without loss of generality.\par
    Recall that $\phi = \angle v_1 m w_1$. We have
    $\phi = \angle v_1 m o + \angle w_1 m o$. Therefore, either $\angle v_1 m w_1$ or $\angle w_1 m o$ must be greater than or equal to $\phi/2$. First, suppose $\angle v_1 m o \geq \phi / 2$. By the Lemma \ref{polar},
    \begin{align*}
        \sin \frac{\angle v_1mo}{2} \geq \frac{\sin \phi}{4} \geq \frac{\{\sin(\angle(s,s')) \}^2}{36} \geq 4 e^{-h/2}
    \end{align*}
    Therefore,
    \begin{align*}
    \cosh \left \{ dist(v_1, o) \right \} &= \cosh(a/2) \cos(h) - \sinh (a/2) \sinh(h) \cos (\angle v_1 m o)\\
    &> \cosh(a/2) \cosh(h) \sin^2 \left(\frac{\angle v_1 m o}{2} \right) \\
    &\geq \cosh(R/2) \frac{1}{2}e^h 16 e^{-h}\\
    &\geq 8\cosh(R/2)
    \end{align*}
    This is a contradiction since $dist(v_1,o)$ must be less than or equal to $R/2$.
    On the other hand, suppose $\angle w_1 m o \geq \phi / 2$. Then,
    \begin{align*}
    \cosh \left \{ dist(w_1, o) \right \} &\geq \cosh(r_0) \cosh(h) \sin^2 \left(\frac{\angle w_1 m o}{2} \right) \\
    &> \frac{1}{8} \cosh(R/2) \frac{1}{2}e^h 16 e^{-h}\\
    &\geq \cosh(R/2)
    \end{align*}
    Again, this is a contradiction.
\end{proof}
\begin{figure}[H]
\centering
\begin{center}
    \includegraphics[width = 9cm]{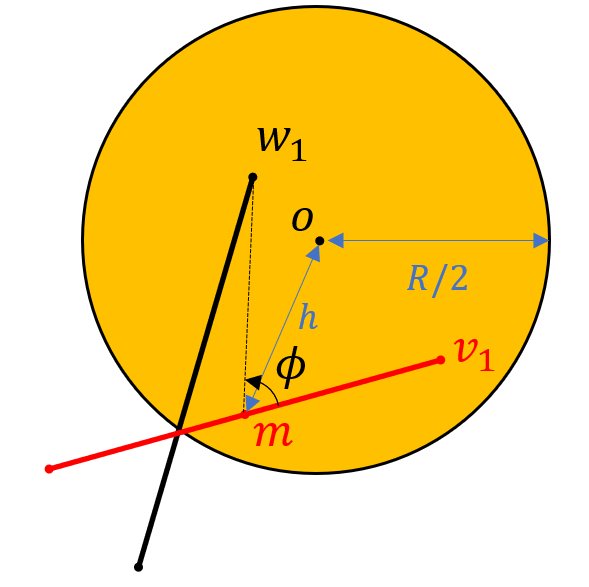}
\end{center}
\caption{The situation of the Lemma \ref{toobigangle}. The yellow region is $B_0(R/2)$. If $\angle(s,s')$ is large, then by the Proposition \ref{polar}, $\phi$ must be large. Thus, either length of $\overline{v_1o}$ or $\overline{w_1o}$ exceeds $R/2$.}
\end{figure}

Using the Lemma \ref{toobigangle}, we divide independent segments into some sets so that segments in each set satisfy either condition.
\begin{itemize}
    \item Angles between them are not so small but are located at the sparse area (See (\ref{Hthree1}))
    \item Angles between them are small but might be located at the dense area (See (\ref{Hthree2}))
\end{itemize}

\begin{lem}
    \label{pigeonplus}
    Let $k$ be a positive integer.
    If there exists a set of $t$ independent segments $S\subseteq \mathcal{S}$, then there exists a segment $s'$,
    a set of independent segments $S' \subseteq \mathcal{S}$ such that,
    \begin{align}
        \label{Hthree1}
        \left \{
        \begin{array}{c}
        |S'| \geq \left \lceil \frac{t}{3k} \right \rceil\\
        s'\in S'\\
        \forall s''\in S'. \angle(s', s'') \leq \pi/k\\
        \forall s''\in S'. \text{Both endpoints of $s''$ are outside of $B_0(R/2)$}
        \end{array}
        \right.
    \end{align}
    or 
    \begin{align}
        \label{Hthree2}
    \left \{
        \begin{array}{c}
        |S'| \geq \left \lceil \frac{t}{3k} \right \rceil\\
        s'\in S'\\
        \forall s''\in S'. \angle(s', s'') \leq 12 e^{-h/4}\pi/k\\
        \forall s''\in S'. \text{Both endpoints of $s''$ are outside of $B_0(\max\{R/2- 3 \log 2-h,0 \}$})\\
    \end{array}
    \right.
    \end{align}
    for some $h \in [0, R]$.
\end{lem}
\begin{proof}
    Take a segment $s_0 \in S$ whose endpoint is in $B_0(R/2)$. If there is none, then we are done with
    the Lemma \ref{hyperpigeon}.
    Let $\theta_{s_0}:= \arcsin (\min \{ 12 e^{-h/4},1 \})$.
    For each index $1\leq i \leq k$, we define 
    \begin{align*}
        S_i &:= \left\{ s\in S: 
        \frac{i-1}{k} \theta_{s_0}\leq \angle(s_0,s) < \frac{i}{k} \theta_{s_0}\right\}\\
        T_i &:= \left\{ s\in S: 
        \theta_{s_0} + \frac{i-1}{k} (\pi - 2\theta_{s_0}) \leq \angle(s_0,s) < \theta_{s_0}+\frac{i}{k} (\pi - 2\theta_{s_0}) \right\}\\
        U_i &:= \left\{ s\in S: 
        \pi - \theta_{s_0} + \frac{i-1}{k} \theta_{s_0} \leq \angle(s_0,s) < \pi - \theta_{s_0}+\frac{i}{k} \theta_{s_0} \right\}
    \end{align*}
    Since $S_i \cup T_i \cup U_i = S$, $\sum_{i=0}^{k-1} |S_i|+|T_i|+|U_i| \geq t$. By the pigeonhole principle, there exists an index $i$
    such that either $S_i$, $U_i$, or $T_i$ has the cardinality greater or equal to $\lceil t/(3k) \rceil$. Let $S'$ be such a set. 
    Consider taking the segment $s'$ in $S'$ so that
        $\angle(s_0,s')=\min_{s''\in S'} \angle(s_0,s'')$.
    If $S'=S_i$ or $S'=U_i$, $\angle(s',s'')\leq \angle(s_0,s'')- \angle(s_0,s') \leq \theta_{s_0}/k \leq 12e^{-h/4} \pi /k$. The third condition of (\ref{Hthree2}) is satisfied. Also, by the Corollary \ref{hyperr0}, endpoints of $s''$ must be outside of a circle whose radius and center are $R/2-3 \log 2$ and $m$. Therefore, the fourth condition of (\ref{Hthree1}) is satisfied. We have proved that $s'$ and $S'$ fulfill (\ref{Hthree2}).
    If $S'=T_i$, then we must have $12e^{-h/4}< 1$ so that $S'(=T_i)$ is not empty. Therefore, $\sin(\angle(s_0,s)) \geq \sin(\theta_{s_0})= 12e^{-h/4}$. By Lemma \ref{toobigangle}, both endpoints of $s'' \in S'$ must be outside of $B_0(R/2)$, and (\ref{Hthree1}) holds.
\end{proof}

\begin{thm}
    Let $\epsilon$ be a positive constant.
    Let $t=n^{(1-\alpha)/3+\epsilon}$. Then,
    \begin{align*}
        \Pr[G_{n,\gamma,C} \text{ has $O_t$ as its vertex-induced subgraph}] \rightarrow 0
    \end{align*}
    as $n \rightarrow \infty$.
\end{thm}
\begin{proof}
    Let $k:=n^{(1-\alpha)/3}$ and $\theta_0:=\pi/k$.
    Suppose $G_{n,r}$ has $O_t$ as its vertex-induced subgraph.
    Apply Lemma \ref{pigeonplus} for $t$ and $k$. As in the proof
    of $\ref{EOt}$, Then there exists a segment $s'$ and a set of independent segments $S'$ such that
    either (\ref{Hthree1}) or (\ref{Hthree2}) holds.
    We first claim that for a certain segment $s'$, the probability that
    there exists a set of independent segments $S'$ such that (\ref{Hthree1}) holds is at most $\exp(-\Omega(n^{\epsilon}))$. By the Corollary \ref{hyperU}, the expected number of segments which is possibly in $S'$ can be bounded by $\rho(R/2) \cdot O(n^2 \theta_0^3)=O(n^{-1-\alpha}) \cdot O(n^2\cdot n^{-(1-\alpha)})=O(1)$. However, the size of $S'$ must be $n^{\epsilon}$. By the Chernoff bound, we have the probability bound. We next claim that for a certain segment $s'$, the probability that
    there exists a set of independent segments $S'$ such that (\ref{Hthree2}) holds is at most $\exp(-\Omega(n^{\epsilon}))$. Again, it is sufficient to prove that the expected number of segments that satisfy the third and the fourth condition of (\ref{Hthree2}) is $O(1)$. Let $\theta_{s_0}:=12e^{-h/4}\pi/k=12e^{-h/4}\theta_0$. Here,
    \begin{align*}
        (\text{The expected value})
        &\leq n\rho(\max\{R-3 \log 2 - h, 0\}) \cdot O(n\theta_{s_0}^3)\\
        &= n \cdot \frac{\alpha}{2\pi}\cdot (\cosh(\alpha R)-1)^{-1} e^{-(1-\alpha)(R/2-3 \log 2-h)} \cdot O(n e^{-3h/4} \theta_0^3)\\
        &= n O(n^{-2\alpha}n^{-1+\alpha} e^{(1-\alpha)h}) \cdot O(ne^{-3h/4} \theta_0^3)\\
        &= O(n^{1-\alpha}e^{-h(3/4-(1-\alpha))} \theta_0^3)\\
        &= O(n^{1-\alpha} \theta_0^3)\ (\because 1/2 < \alpha < 1)\\
        &= O(1)
    \end{align*}
    Finally, we apply a union bound over all segments and obtain
    \begin{align*}
        \Pr[G_{n,\gamma,C} \text{ has $O_t$ as its vertex-induced subgraph}] \leq \frac{n(n-1)}{2} \cdot 2\exp(-\Omega(n^{\epsilon}))
    \end{align*}
    which goes to 0 as $n \rightarrow \infty$.
\end{proof}

\begin{thm}
    There exist a positive constant $C'$ such that for all $\epsilon > 0$,
    \begin{align*}
        \Pr[\mathcal{M}(G_{n,\gamma,C})\leq \exp(C'n^{(1-\alpha)/3+\epsilon})]\rightarrow 1
    \end{align*}
    as $n\rightarrow \infty$

\end{thm}

\section{Conclusions}
In this paper, we started with the question: why is the number of maximal cliques
small on real-world networks? To give an explanation to it, we consider the number of 
maximal cliques in two models of real-world networks: Euclidean random geometric graphs and
hyperbolic random graphs. To bound the number, we focus on vertex-induced subgraphs, which are isomorphic to $O_t$. Similar geometric and probabilistic 
techniques apply to both random graphs.
For future works, we are interested in whether the property of $O_t$ has positive effects on
clique enumeration algorithms on real-world graphs. Also, the generalization to
a higher dimensional space is an open question.

\begin{table}[b]
\centering
\begin{tabular}{lrrrrr}
\hline
Graph &          $|V|$ &           $|E|$ &   $\tau$  \\
\hline
amazon0601 & 403394 & 2443408 & 5\\
as-skitter & 1696415 & 11095298 & 15\\
ca-AstroPh & 18772 & 198050 & 7\\
com-dblp & 317080 & 1049866 & 4\\
com-youtube & 1134890 & 2987624 & 7\\
email-Enron & 36692 & 183831 & 7\\
email-EuAll & 265214 & 364481 & 8\\
facebook\_combined & 4039 & 88234 & 19\\
hrg & 180705 & 968205 & 4\\
loc-gowalla\_edges & 196591 & 950327 & 8\\
soc-buzznet & 101164 & 2763067 & 14\\
soc-digg & 770800 & 5907133 & 17\\
soc-Epinions1 & 75879 & 405740 & 10\\
web-Google & 875713 & 4322051 & 5\\
web-NotreDame & 325729 & 1090108 & 5\\
web-Stanford & 281903 & 1992636 & 5\\
wiki-Talk & 2394385 & 4659565 & 13\\
wiki-topcats & 1791489 & 25444207 & 10\\
wiki-Vote & 7115 & 100762 & 8\\
\hline
\end{tabular}
\caption{Upper bound of $\tau$ (i.e. maximum $t$
such that a graph contains $O_t$ as its vertex-induced subgraph) on SNAP graph datasets}
\end{table}

\end{document}